\documentclass{llncs}[orivec]
\usepackage{silence}
\usepackage{amssymb}
\usepackage{latexsym}
\usepackage{amsfonts,amsmath}
\usepackage{mathabx}
\usepackage{array,longtable}

\usepackage{pgfplots}
\usepackage{rotating}
\pgfplotsset{compat=1.15}
\usepackage{mathrsfs}
\usetikzlibrary{arrows}

\usepackage{stmaryrd}
\SetSymbolFont{stmry}{bold}{U}{stmry}{m}{n}
\usepackage{alltt}
\usepackage{epic,epsfig}
\usepackage{graphics}
\usepackage{algorithm}
\usepackage{algpseudocode}
\usepackage{algpseudocodex}
\usepackage{algorithm}
\usepackage{xcolor}
\usepackage{float}
\usepackage{tikz}
\usepackage{stackrel}
\usetikzlibrary{arrows.meta, positioning}

\newcommand{\lgu}{| \hspace{-0.18mm}u \hspace{-0.18mm}|}

\newcommand{\lgN}{| \hspace{-0.18mm}N \hspace{-0.18mm}|}
\newcommand{\lgx}[1]{| \hspace{-0.18mm}{#1} \hspace{-0.18mm}|}
\newcommand{\spc}{sp}

\newcommand{\voidsequence}{(\,)}

\newcommand{\sequencef}[3]{(#1_i)_{i\in[#2:#3]}}
\newcommand{\sequenceo}[3]{(#1_i)_{i\in[#2:#3[}}

\newcommand{\interf}[2]{[#1\!:\!#2]}
\newcommand{\intero}[2]{[#1\!:\!#2[}
\newcommand{\pair}[2]{\langle #1,#2\rangle}
\newcommand{\voidpair}{\pair{\voidsequence}{\voidsequence}}

\newcommand{\win}[3]{\langle (#1_i)_{i\in[0:#3]}, (#2_i)_{i\in[0:#3[}\rangle } 

\newcommand{\infwin}[2]{\langle (#1_i)_{0\leq i}, (#2_i)_{0\leq i}\rangle }
\newcommand{\partialwin}[4]{\langle (#2_i)_{i\in[#1:#4]}, (#3_i)_{i\in[#1:#4[}\rangle } 
\newcommand{\shrinkwin}[4]{\langle (#2_i)_{i\in[#1:#4]}, (#3_i[0])_{i\in[#1:#4[}\rangle }
\newcommand{\longwin}[4]{\langle (#1_i)_{i\in[0:#3]}, (#2_i)_{i\in[0:#3[}\rangle } 

\newcommand{\lbr}{[}
\newcommand{\rbr}{]}

\newcommand{\sizew}{n}

\newcommand{\card}{\mathrm{Card}}
\newcommand{\Bi}[1]{\ensuremath{\Box_{#1}}}
\newcommand{\Bim}[1]{\ensuremath{\Box_{#1}^{\,\mhyphen}}}
\newcommand{\Ri}[1]{\ensuremath{R_{#1}}}
\newcommand{\Kde}{\ensuremath{\mathbf{KDe}}}
\newcommand{\Kpi}{\ensuremath{\mathbf{KDe}(\Pi)}}

\newcommand{\K}{\ensuremath{\mathbf{K}}}
\newcommand{\nat}{\mathbb{N}}

\newcommand{\pipe}{\hspace{-0.21mm}|\hspace{-0.21mm}}

\newcommand{\inc}{\subseteq}
\newcommand{\uplim}[1]{\ensuremath{\chi(#1)}}
\newcommand{\wpp}{\ensuremath{w^+}}

\newcommand{\Wpp}{\ensuremath{\tilde{W}}}

\newcommand{\nextW}{\mbox{\texttt{NextW}}}
\newcommand{\satW}{\mbox{\texttt{SatW}}}
\newcommand{\chooseW}{\mbox{\texttt{ChooseW}}}
\newcommand{\chooseCCS}{\mbox{\texttt{ChooseCCS}}}
\newcommand{\sat}{\mbox{\texttt{Sat}}}

\newcommand{\CCS}{\mbox{\texttt{CCS}}}
\newcommand{\SF}{\mbox{\texttt{SF}}}
\newcommand{\CSF}{\mbox{\texttt{CSF}}}

\newcommand{\true}{\mbox{\texttt{True}}}

\newcommand{\all}{\mbox{\texttt{all}}}
\newcommand{\algand}{\mbox{\texttt{and}}}

\usepackage{amsmath}

\def\PSPACE{\mathbf{PSPACE}}
\def\card{\mathtt{Card}}
\def\NPSPACE{\mathbf{NPSPACE}}
\def\EXPSPACE{\mathbf{EXPSPACE}}
\def\EXPTIME{\mathbf{EXPTIME}}
\def\NEXPTIME{\mathbf{NEXPTIME}}
\def\coNEXPTIME{\mathbf{coNEXPTIME}}
\def\CPL{\mathbf{CPL}}
\def\K{\mathbf{K}}
\def\At{\mathbf{At}}
\def\Fo{\mathbf{Fo}}
\def\Axiom{\mathbf{A}}
\def\Rule{\mathbf{R}}
\mathchardef\mhyphen="2D
\begin{document}
\title{Recursive windows for grammar logics of bounded density}
\author{Olivier Gasquet\\
olivier.gasquet@irit.fr}
\institute{
Institut de recherche en informatique de Toulouse (IRIT)
\\
CNRS-INPT-UT
}
\maketitle
\begin{abstract}
We introduce the family of multi-modal logics of bounded density and with a tableau-like approach using finite \emph{windows} which were introduced in \cite{BalGasq25} and that we generalize to recursive windows. We prove that their satisfiability problem is {\bfseries PSPACE}-complete. As a side-product, the monomodal logic of density is shown to be in para-{\bfseries PSPACE}. 
\end{abstract}
%
%
%
\section*{Introduction}
\paragraph{Context} The satisfiability problem of the modal logics defined by \emph{grammar axioms} of the form $\Diamond_{a_1}\ldots\Diamond_{a_m} p\rightarrow \Diamond_{b_1}\ldots\Diamond_{b_n} p$, is known to be undecidable in general~\cite{FarinasdelCerro1988-FARGL}.
For ``transitivity like'' axioms, extensive work has been done in e.g.\ \cite{Baldoni2}, \cite{Demri00} and \cite{Demri01}. They characterize a class of grammar logics with axioms of the form 
$\Diamond_{a_1}\ldots\Diamond_{a_m} p\rightarrow \Diamond_{a} p$, which are also known as \emph{reduction principle} (\cite{VBen76},\cite{CS94}) as they semantically correspond to the existence of shorter paths in the associated Kripke frames. These axioms, expressed as grammar rules $a\rightarrow^*_{\cal G} a_1\cdots a_m$, correspond to a linear grammar ${\cal G}$: left linear if $a=a_1$ or right linear if $a=a_m$. In the former case, most of the time they have an $\EXPTIME$-complete satisfiability problem, and a $\PSPACE$-complete one in the latter case when ${\cal G}$ is finite (there are $\EXPTIME$-complete ones when it is infinite). There are also some high lower bound for monomodal such logics as finite model property in \cite{Zakha} and $\EXPSPACE$-satisfiability in the more recently \cite{Lyon24}. \\
For \emph{grammar logics} containing what we call ``weak-density'' axioms: $\Diamond_{a} p\rightarrow \Diamond_{a_1}\ldots\Diamond_{a_m} p$, and while their decidability is easy to establish by means of a least filtration as defined in \cite{blackburn01}, giving a $\NEXPTIME$ lower bound for satisfiability, not much is known about their precise complexity, and we will prove $\PSPACE$-completeness for a particular fragment of them. \\

\paragraph{Motivation} Surely one may think these bounded-dense logics may appear somehow artificial, but in a temporal interpretation, it could make sense to take for granted that time is not infinitely decomposable and that if between two instants there is another one, this reasoning cannot be applied infinitely. These logics for which we prove $\PSPACE$-completeness aims at covering such a bounded notion of density (hence the title) by fixing an a priori bound to the possible depth of time-density not depending on the formulas one want to check. \\
On another hand, in section \ref{KDe}, we also prove, using the same apparatus that the monomal logic of density is in para-$\PSPACE$: for any given formula density is bounded its the modal depth. Hence, all formulas of fixed depth can be checked for satisfiability in polynomial space \emph{w.r.t.\ this modal depth}. 

\paragraph{Our result}%
In this paper, we study  a  subclass of such logics and  their complexity. They are defined by several density-like axioms of the form $\Diamond_i p\rightarrow \Diamond_i\Diamond_{i+1}  p$ for $i\in\{0,\cdots,\pi\}$ that we call them \emph{$\Pi$-dense logics} because they correspond to a bounded form of density of the accessibility relation of their models. If we again consider that these axioms may be expressed as grammar rules, but the other way round i.e.\ axiom $\Diamond_{a} p\rightarrow \Diamond_{a_1}\ldots\Diamond_{a_m} p$ correspond to rule $a\rightarrow^*_{\cal G} a_1\cdots a_m$, then a word $w=a_{w_1}.\cdots .a_{w_n}$ can be obtained from $a$ would imply that in Kripke frame for this logic $R_a\subseteq R_{a_{w_1}}\circ\cdots \circ R_{a_{w_n}}$. If this grammar is recursive (i.e.\ $a\in\{a_1,\cdots,a_m\}$, then such property implies the existence of infinitely long path between possible worlds. The easy case is when $G$ is right-linear since in this case, the intermediary worlds between, say, $u$ and $v$, are farther and farther from $u$ and at some point they are too far for influencing truth-value of formulas in $v$. The difficult case, which we investigate is left-linearity, i.e.\ when $a=a_1$. Then intermediary worlds are all equally close to $u$. Our idea to deal with this was to exploit a loop argument in the style of what is done in the seminal Ladner's paper \cite{Ladner77} and thoroughly used elsewhere. But, our loops are in the width, not in the depth. To our knowledge, this is a new kind of argument. This idea was successfully applied in \cite{BalGasq25} in which the complexity of the satisfiability problem for the bimodal logic of weak density $\K_1\oplus\K_2$ $+\Diamond_1 p\rightarrow \Diamond_1 \Diamond_2p$ has been thus proved to be $\PSPACE$-complete. The approach made use of so-called \emph{windows} which are finite sliding partial view of a tableau structure. In this previous paper, windows were just sequences of sets. Here, in order to address the case of several such left-linear axioms, we extend them to \emph{recursive windows}, finite structures made of such sets. This generalization provides the adequate tool for proving that the satisfiability problem for what we call $\Pi$-dense logics is $\PSPACE$-complete in general.  

\paragraph{Our contribution} Beyond the mere result of complexity for $\Pi$-dense logics, our contribution also consists in having designed a new tool for exploring complexity of some modal logics, namely \emph{recursive windows} and by exploiting loops \emph{in the width}. One may ask whether there is a connection between our windows and so-called mosaics of \cite{Marx00} that were first introduced in \cite{Nemeti95}. In fact, not really, mosaics are small pieces of model (two of a few worlds) aimed at being put together to form a model, and even if windows may be viewed as a kind of overlapping mosaics, membership in $\PSPACE$ is mostly due to the loop argument which is an original important feature. 

\paragraph{Plan of the paper}
Sections \ref{pi-density} is devoted to the definition of $\Pi$-dense logics, axiomatically and semantically. Then after some basic definitions needed in section \ref{settings}, we define the notion of window in section \ref{windows} and show some important properties that will be needed in the sequel. Right after, section \ref{algorithm} describes a non-deterministic algorithm for testing the satisfiability of some finite set of formulas. It is followed by the proof of its soundness and completeness in section \ref{section:complexity:of:Kpi} together with the proof that it runs in polynomial space. We conclude by identifying some open problems that could be tackled by using our windows. 
\section{Logics of $\Pi$-density}\label{pi-density}
Let $\pi\in\nat$. The $\Pi$-dense modal logic $\Kpi$ is the direct product of $\pi+1$ versions of $\K$ with $\pi$ additional density axioms: $\bigoplus_{0\leq k\leq\pi} \K_i$+$\bigwedge_{0\leq k<\pi} \Bi{i}\Bi{i+1} p\rightarrow \Bi{i} p$.
Obviously, \Kpi\ is a conservative extension of ordinary modal logic $\K$.
Hence the satisfiability problem for \Kpi\ is $\PSPACE$-hard\label{pspace-hardness}. In the sequel we prove that it is in $\PSPACE$.

\paragraph{Syntax}
Let $\Pi$ be the set of natural numbers from 0 to $\pi$, and $\Pi^-$ be the set of natural numbers from 0 to $\pi-1$. \\
Let $\At$ be the set of all atoms $(p,q,\ldots)$.
The set $\Fo$ of all formulas $(\phi,\psi,\ldots)$ is defined by
$$\phi:=p\mid\bot\mid\neg\phi\mid(\phi\wedge\phi)\mid\Bi{i}\phi$$
where $p$ ranges over $\At$ and $i$ in $\Pi$.
As before, we follow the standard rules for omission of the parentheses, we use the standard abbreviations for the Boolean connectives $\top$, $\vee$ and $\rightarrow$ and for all formulas $\phi$, $d(\phi)$ denotes the modal depth of $\phi$, and $\pipe\phi\pipe$ denotes the number of occurrences of symbols in $\phi$.
For all formulas $\phi$, we write $\Diamond_{i}\phi$ as an abbreviation  of $\neg\Bi{i}\neg\phi$.

\paragraph{Semantics}
A {\em frame}\/ is a pair $(S,(\Ri{i})_{i\in\Pi})$ where $S$ is a nonempty set and for each $i\in\Pi$, $\Ri{i}$ is a binary relation on $S$, i.e.\ $R_i\subseteq W^2$.
A frame $(S,(\Ri{i})_{i\in\Pi})$ is {\em $\Pi$-dense}\/ if for all $s,t\in S$ and $i\in\Pi^-$, if $sR_{i}t$ then there exists $u\in S$ such that $sR_{i}u$ and $uR_{i+1}t$.
A {\em valuation on a frame $(S,(\Ri{i})_{i\in\Pi})$}\/ is a function $V\ :\ \At\longrightarrow \wp(u)$.
A {\em model}\/ is a $3$-tuple $(S,(\Ri{i})_{i\in\Pi},V)$ consisting of a frame $(S,(\Ri{i})_{i\in\Pi})$ denoted by ${\cal F}(M)$, and a valuation $V$ on that frame. 
A {\em model based on the frame $(S,(\Ri{i})_{i\in\Pi})$}\/ is a model of the form $(S,(\Ri{i})_{i\in\Pi},V)$.
With respect to a model $M=(S,(\Ri{i})_{i\in\Pi},V)$, for all $x\in S$ and for all formulas $\phi$, the {\em satisfiability of $\phi$ at $x$ in $M$}\/ (in symbols $M,x\models\phi$) is inductively defined as usual.
In particular,
\begin{itemize}
\item For all $i\in\Pi$, $x\models\Bi{i}\phi$ if and only if for all $y\in S$, if $x\Ri{i}y$ then $y\models\phi$.
\end{itemize}
As a result,
\begin{itemize}
\item For all $i\in\Pi$, $x\models\Diamond_{i}\phi$ if and only if there exists $y\in S$ such that $x\Ri{i}y$ and $y\models\phi$.
\end{itemize}
A formula $\phi$ is {\em true in a model $(S,(\Ri{i})_{i\in\Pi},V)$}\/ (in symbols $(S,(\Ri{i})_{i\in\Pi},V)\models\phi$) if for all $s\in s$, $s\models\phi$.
A formula $\phi$ is {\em valid in a frame $(S,(\Ri{i})_{i\in\Pi})$}\/ (in symbols $(S,(\Ri{i})_{i\in\Pi})\models\phi$) if for all models $(S,(\Ri{i})_{i\in\Pi},V)$ based on $(S,(\Ri{i})_{i\in\Pi})$, $(S,(\Ri{i})_{i\in\Pi},V)\models\phi$.
A formula $\phi$ is {\em valid in a class ${\mathcal F}$ of frames}\/ (in symbols ${\mathcal F}\models\phi$) if for all frames $(S,(\Ri{i})_{i\in\Pi})$ in ${\mathcal F}$, $(S,(\Ri{i})_{i\in\Pi})\models\phi$.
\paragraph{A decision problem}
Let $DP_{\Kpi}$ be the following decision problem:
\begin{description}
\item[input:] a formula $\phi$,
\item[output:] determine whether $\phi$ is valid in the class of all $\Pi$-dense frames.
\end{description}
Using the fact that the least filtration of a $\Pi$-dense model is $\Pi$-dense, one may readily prove that $DP_{\Kpi}$ is in $\coNEXPTIME$.
We will prove in what follows that $DP_{\Kpi}$ is in $\PSPACE$.
\paragraph{Axiomatization}
In our language, a {\em modal logic}\/ is a set of formulas closed under uniform substitution, containing the standard axioms of $\CPL$, closed under the standard inference rules of $\CPL$, containing the axioms $(\Axiom1_{i})$ $\Bi{i} \top$ and $(\Axiom2_{i})$ $\Bi{i} p\wedge\Bi{i} q\rightarrow \Bi{i}(p\wedge q)$ (for $i\in\Pi$)
and closed under the inference rules

\[(\Rule1_{i}) \frac{p\rightarrow q}{\Bi{i}p\rightarrow \Bi{i}q} \mbox{(for $i\in\Pi$)}\]
Let \Kpi\ be the least modal logic containing the formulas $\Bi{i}\Bi{i+1}p\rightarrow \Bi{i}p$ for $i\in\Pi^-$. 
As is well-known, \Kpi\ is equal to the set of all formulas $\phi$ such that $\phi$ is valid in the class of all $\Pi$-dense frames. 
This can be proved by using the so-called canonical model construction. Moreover, as the reader may easily verify, the class of $\Pi$-dense frames is closed under disjoint union ($\sqcup)$. Fig.~\ref{fig:frame} shows what these frames look like locally. 
\begin{figure}[!ht]
\begin{centering}
  \fbox{
  \begin{tikzpicture}[scale=0.5]
\coordinate (u) at (4,8) ;
\coordinate (v0) at (16,8) ;
\coordinate (v1) at (12,7) ;
\coordinate (v2) at (8,6);
\coordinate (v3) at (6,5.5) ;

\coordinate (v03) at (11,5.5);
\coordinate (v02) at (12,6);
\coordinate (v01) at (14,7) ;
\coordinate (v12) at (8,5) ;
\coordinate (v11) at (10,6) ;
\coordinate (v13) at (7,4.5);
\coordinate (v001) at (15,7) ;
\coordinate (v002) at (14,6) ;
\coordinate (v003) at (13.5,5.5) ;
\coordinate (v013) at (11,4) ;
\coordinate (v113) at (7.5,3.5) ;

\coordinate (v011) at (13,6) ;
\coordinate (v012) at (12,5) ;
\coordinate (v101) at (11,6) ;
\coordinate (v102) at (10,5) ;
\coordinate (v103) at (9.5,4.5);
\coordinate (v111) at (9,5) ;
\coordinate (v112) at (8,4) ;
\draw (u) node[above]{$u$};
\draw (v0) node[above]{$v$};
\draw  (u) -- (v0) node [midway,sloped]{$\circ$};
\draw  (u) -- (v1) node [midway,sloped]{$\circ$};
\draw  (u) -- (v2) node [midway,sloped]{$\circ$};
\draw (v2) -- (v1) node [midway,sloped]{{\small $\square$}} -- (v0) node [midway,sloped]{{\small $\square$}} ;
\draw [dashed] (v3) -- (v2) ;
\draw (v2) -- (v12) node [midway]{{\small $\square$}};
\draw (v2) -- (v11) node [midway]{{\small $\square$}};
\draw  (v12) -- (v112);
\draw  (v12) -- (v111);
\draw  (v12) -- (v11) -- (v1);
\draw [densely dotted] (v112) -- (v111) -- (v11);
\draw  (v11) -- (v102);
\draw (v11) -- (v101);
\draw [densely dotted] (v102) -- (v101) -- (v1);
\draw (v2) -- (v1) -- (v0);
\draw (v1) -- (v02) node [midway]{{\small $\square$}};
\draw (v1) -- (v01) node [midway]{{\small $\square$}};
\draw (v02) -- (v01) -- (v0);
\draw [dashed] (v03) -- (v02);
\draw [dashed] (v13) -- (v12);
\draw [dashed] (v103) -- (v102);
\draw [dashed] (v113) -- (v112);
\draw [dashed] (v013) -- (v012);
\draw [dashed] (v003) -- (v002);

\draw  (v01) -- (v002);
\draw (v01) -- (v001);
\draw  (v02) -- (v012);
\draw (v02) -- (v011);
\draw  [densely dotted] (v012) -- (v011) -- (v01)  ;
\draw  [densely dotted] (v002) -- (v001) -- (v0);

\end{tikzpicture}
}
\caption{Local structure of a $\Pi$-dense frame between $u$ and $v$, with $Pi=\{0\cdots 3\}$.\\ $\circ$: $R_0$, {\small $\square$}: $R_1$, no label: $R_2$, dotted: $R_3$. Edges are left-to-right, or else top-to-bottom.}\label{fig:frame}
\end{centering}
\end{figure}
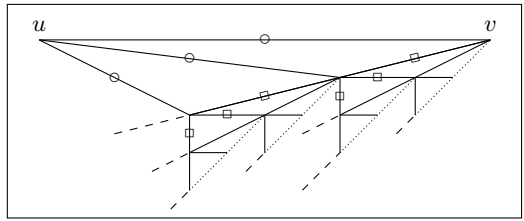

\section{Basic settings}\label{settings}
Let us introduce some notations: 
\begin{itemize}
    \item Subtraction on integers: $n\dotdiv m=\max(0,n-m)$
    \item intervals of $\nat$: 
    \item [] \indent $\interf{a}{b}
    =\{a,a+1,\cdots,b\}$ (if $b<a$ then $\interf{a}{b}=\emptyset$)
    \item [] \indent $\intero{a}{b}=\{a,a+1,\cdots,b-1\}$ (if $b\leq a$ then $\intero{a}{b}=\emptyset$)
    \item In the context of sequences of sets of formulas, $\voidsequence$ denotes the empty sequence
    and $\{\_\}$ denotes the set of members of a sequence: $y\in\{(y_i)_{i\in I}\}$ iff $\exists i\in I\colon y=y_i$
    \item If $x\in X$ and $y\in Y$, then $\pair{x}{y}$ denotes the corresponding element of $X\times Y$
    \item Membership is extended to cartesian product: $\{\pair{X}{Y}\}=\{X\}\cup \{Y\}$
    \item For $s$ a finite set of formulas $d(s)=\max\{d(\phi):\ \phi\in s\}$ and $\pipe s\pipe=\Sigma\{\pipe\phi\pipe:\ \phi\in s\}$
    \item For $s$ finite set of formulas, $\Bim{i}(s)=\{\phi\colon \Bi{i}\phi \in s \}$ (notice that $d(\Bim{i}(s))\leq d(s)\dotdiv 1$ and that $\Bim{i}(s\cup s')=\Bim{i}(s)\cup\Bim{i}(s')$). 
\end{itemize} 
\paragraph{Subsets of formulas\\}
Let $\CSF(s)$ (for Classical SubFormulas) be the least set $s'$ of formulas such that for all formulas $\phi,\psi$,
\begin{itemize}
\item $s\subseteq s'$,
\item if $\phi \wedge \psi\in s'$ then $\phi \in s'$ and $\psi\in s'$,
\item if $\neg (\phi \wedge \psi)\in s'$ then $\neg \phi\in s'$ and $\neg \psi\in s'$,
\item if $\neg \phi\in s'$ then $\phi \in s'$.
\end{itemize}
In other respect, $\SF(s)$ (for SubFormulas) is the least set $s'$ of formulas s. th.\ for all formulas $\phi,\psi$,
\begin{itemize}
\item $s\subseteq s'$,
\item if $\phi \wedge \psi\in s'$ then $\phi \in s'$ and $\psi\in s'$,
\item if $\neg (\phi \wedge \psi)\in s'$ then $\neg \phi\in s'$ and $\neg \psi\in s'$,
\item if $\neg \phi\in s'$ then $\phi \in s'$,
\item if $\Bi{i} \phi\in s'$ then $\phi \in s'$ (for $i\in\Pi$),
\item if $\neg\Bi{i} \phi\in s'$ then $\neg\phi \in s'$ (for $i\in\Pi$). 
\end{itemize}
It is well-known that $\lgx{\SF(u)}\leq c_{\mathrm{sf}}.\lgu$ for some constant $c_{\mathrm{sf}}>0$. \\

\noindent For any finite set $s$ of formulas, let $\CCS(s)$ be the set of all finite sets $u$ of formulas such that $s\subseteq u\subseteq \CSF(s)$ which are Classically Consistent and classically Saturated, i.e.\ for all formulas $\phi,\psi$:
\begin{itemize}
\item if $\phi \wedge \psi\in u$ then $\phi \in u$ and $\psi\in u$,
\item if $\neg (\phi \wedge \psi)\in u$ then $\neg \phi\in u$ or $\neg \psi\in u$,
\item if $\neg \neg \phi\in u$ then $\phi \in u$,
\item $\bot\not\in u$,
\item if $\neg \phi\in u$ then $\phi\not\in u$.
\end{itemize}
For finite consistent set $s$ of formulas, the elements of $\CCS(s)$ are in fact simply unsigned saturated open branches for tableaux \footnote{Originating in the works of Evert W. Beth, probably in \cite{Beth} as pointed by M. Guillaume in \cite{Guillaume} and made famous with the work of M. Fitting in \cite{Fitting85} for his extension to modal logics}.
As a result, for all finite sets $s$ of formulas, an element of $\CCS(s)$ is called a {\em consistent classical saturation (CCS) of $s$.}
As the reader may easily verify, for all finite sets $u, w$ of formulas, if $u\in \CCS(s)$ then $d(u)=d(s)$ and $\CCS(u)=\{u\}$.
Moreover, there exists an integer $c_{\mathrm{ccs}}$ such that for all finite sets $u, s$ of formulas, if $u\in \CCS(s)$ then $\pipe u\pipe\leq c_{\mathrm{sf}}.\pipe s\pipe$.
\begin{proposition} [Properties of {\CCS}s]\label{prop-CCS}
For all finite sets $u, v, s, s_1, s_2$ of formulas,
\begin{enumerate}
\item\label{One} if $u\in \CCS(s\cup v)$ and $v\in \CCS(s')$ then $u\in \CCS(s\cup s')$,
\item\label{Two} if $u\in \CCS(s\cup s')$ then it exists $v\in \CCS(s)$ and $v'\in\CCS(s')$ s.th.\ \ $v\cup v'=u$,
\item\label{Three} if $u\in \CCS(s'\cup v)$ and $v$ is a \CCS\ then it  exists $v'\in \CCS(s')$ s.th.\ $v\cup v' = u$,
\item\label{Four} if $u\in \CCS(s'\cup v)$ and $v$ is a \CCS\ then $d(u\setminus v)\leq d(s')$,
\item\label{Five} if $u$ is true at a world $x\in s$ of a \Kpi-model $M=(S,R_{a},R_{b},V)$, then 
the set $\SF(u)\cap \{\phi\colon M,x\models \phi\}$ is in $\CCS(u)$.
\end{enumerate}
\end{proposition}
\begin{proof}
Item~(\ref{One}) is an immediate consequence of the properties of classical open branches of tableaux.
As for Item~(\ref{Two}), take $v=s\cap\CSF(u)$ and $v'=s'\cap\CSF(v)$.
Item~(\ref{Three}) follows from Item~(\ref{Two}).
Concerning Item~(\ref{Four}), if $u\in \CCS(s'\cup v)$ then by Item~(\ref{Three}), there exists $v'\in \CCS(s')$ and $v\cup v' = u$.
Therefore, $u\setminus v\subseteq v'$ and $d(u\setminus v)\leq d(v')=d(s')$.
Finally, about Item~(\ref{Five}), the reader may easily verify it 
by applying the definition clauses of $\models$. 
\end{proof}
\section{Recursive windows}\label{windows}
Recursive windows (or just windows) are polynomialy large pieces of information sufficient to explore the possibility of the existence of a model and to ensure $\Kpi$-satisfiability. But, in order to guarantee completeness, we must be able to expand them to piece of model, to this aim we need new, bigger objects (up to exponentialy big) or even infinitely big. All in all, we will call them all windows as well.  \\
In this section, we will use $\lambda$ to denote either an integer function on \CCS\ $u$, with $\lambda(u)\geq d(u)$ and $\lambda$ is increasing w.r.t.\ $d$ (i.e.\ $d(u)>d(v)\rightarrow \lambda(u)>\lambda(v)$), or $\lambda$ may be the constant function $\lambda(u)=\infty$ for all $u$. 
\paragraph{Windows\\} 
Let $u$, $v_0$ be two \CCS, $k\in \Pi$ and $n\geq d(u)$. 
A $(k,n,\lambda)$-window for $(u,v_0)$ denoted by $W$ is a pair $\pair{\cal V}{\cal W}$ with:
\begin{itemize}
        \item if $k=\pi$, $\pair{\cal V}{\cal W}=\voidpair$ (empty window)
        \item if $k<\pi$, $\pair{\cal V}{\cal W}$ is a pair of non-empty sequences with: 
        \begin{itemize}
            \item ${\cal V}=\sequencef{v}{0}{\sizew}$ is a sequence of formulas such that:
        \begin{itemize}   
            \item for each $i\in\intero{0}{\sizew}\colon v_i\in\CCS(\Bim{k}(u)\cup\Bim{k+1}(v_{i+1}))$
            \item if $\sizew\neq \infty$ then $v_{\sizew}\in \CCS(\Bim{k}(u))$
        \end{itemize}
        \item ${\cal W}=\sequenceo{W}{0}{\sizew}$ is a sequence of $(k+1,\lambda(v_{i+1}),\lambda)$-windows for $(v_{i+1},v_i)$, called a subwindow of $W$. 
        \end{itemize}
\end{itemize}
The reader could wonder why we need $n$ and $\lambda$. Number $n$ can be thought of as the length of the window (number of nodes at the top level), and in the course of expending a window, as we will see below, we will need to consider window whose length is different (and longer) than that of its subwindows. Also, according to the case, we need subwindows of (linear) length and of (exponential) length, thus the need of a function $\lambda$ which will essentially be either $d$ or $\chi$. 
We inductively define the set $\{W\}$ of members of a window by: 
\begin{itemize}
    \item $\{\voidpair\}=\emptyset$
    \item $\{\win{v}{W}{\sizew}\}=\{(v_i)_{i\in[0:\sizew]}\}\cup \bigcup_{i\in[0:\sizew[}\{W_i\}$
\end{itemize}

\noindent N.B.\ for any set $v\in \{W\}$, we have: $\lgx{v}\leq c_{\mathrm{sf}}.\lgu$ and $d(v)< d(u)$. 
\paragraph{Partial windows\\} 
Let $T=\win{v}{W}{\sizew}$ be a $(k,\sizew,\lambda)$-window for $(u,v_0)$ as above. Let $0\leq a<b\leq \sizew$, then 
$\partialwin{a}{v}{W}{b}$ is a $(k,b-a+1,\lambda)$-partial window of $W$. Of course, $W$ is a partial window of $W$ (take $a=0$ and $b=\sizew$). 


\paragraph{Pointwise $k$-inclusion of partial windows\\} 
Let $u$, $v^1_0$ and $v^2_0$ be \CCS, let $k\in \Pi$. Let $W^1$ be a $(k,\sizew,\lambda)$-window for $(u,v^1_0)$ and let $W^2$ be a $(k,\sizew,\lambda)$-window for $(u,v^2_0)$. Let $0\leq a<b\leq l$ and let $SW_1$ and $SW_2$ be two $(k,b-a+1,\lambda)$-partial windows of $W_1$ and $W_2$, we define the pointwise $k$-inclusion of the partial windows $SW_1\sqsubseteq^k SW_2$ by:
\begin{itemize}
        \item $\voidpair\sqsubseteq^{\pi} \voidpair$
        \item if $k<\pi\colon \partialwin{a}{v^1}{W^1}{b}\sqsubseteq^k \partialwin{a}{v^2}{W^2}{b}$ iff:
        \begin{itemize}
            \item for $i\in\intero{a}{b}\colon v^2_i\in\CCS(v^1_i\cup\Bim{k+1}(v^2_{i+1}))$
            \item and for $i\in\intero{a}{b}\colon W^1_i\sqsubseteq^{k+1} W^2_i$ 
            \item [] N.B.\ if $d(u)=1$ then for $i\in\intero{a}{b}\colon d(\Bim{k+1}(v^2_{i}))=0$ and hence $v^2_{i}=v^1_{i}$.
        \end{itemize}
\end{itemize}

\paragraph{Continuations of windows\\}
Let $u$, $v^1_0$ and $v^2_0$ be \CCS, let $k\in \Pi^-$, and $d(u)\geq 1$. Let $W^1$ be the $(k,\sizew,\lambda)$-window for $(u,v^1_0)$ given by: $\win{v^1}{W}{\sizew}$  and $W^2$ be the $(k,\sizew,\lambda)$-window for $(u,v^2_0)$ given by: $ \win{v^2}{W}{\sizew}$. We say that $W_2$ is a $k$-continuation of $W_1$ iff 
$$\partialwin{1}{v^1}{W^1}{\sizew} \sqsubseteq^k \partialwin{0}{v^2}{W^2}{\sizew-1}$$
(beware of indexes: the ``end'' of $W_1$ is included ($\sqsubseteq^k$) in the ``beginning'' of $W_2$). 

\begin{lemma}\label{k2kplusone-window}
    Let $u$, $v_0$ be two \CCS, let $k\in \Pi^-$, and $d(u)\geq 1$. Let $W_1$ and $W_2$ be two $(k,\sizew,\lambda)$-windows for $(u,v_0)$ with $W_1=\win{v^1}{W^1}{\sizew}$
     and $W_2=\win{v^2}{W^2}{\sizew}$, and suppose $W_2$ is a $k$-continuation of $W_1$, then with
     \begin{itemize}
         \item $(z_0,D_0)=(v^1_0,W^1_0)$
         \item $z_{\sizew+1}=v^2_{\sizew}$
         \item and for $1<z\leq \sizew\colon (z_i,D_i)=(v^2_{i-1},W^2_{i-1})$
     \end{itemize}
          we have $\win{z}{D}{\sizew+1}$ is a $(k,\sizew+1,\lambda)$-window for $(u,v_0)$. 
\end{lemma} 
\begin{proof}
First we need to prove the following proposition:
\begin{proposition}\label{proposition-degree}
For $1\leq i\leq \sizew\colon d(v^2_{i-1}\setminus v^1_i) \leq d(u)+i\dotdiv (\sizew+1)$
by descending induction on $i\in\{1,\ldots,\sizew\}$: either $i=\sizew$, or $i<\sizew$.
In the former case, $v^2_{\sizew-1}\in \CCS(v^2_{\sizew}\cup v^1_{\sizew})$.
Since $v^2_{\sizew}\in \CCS(\Bim{k}(u))$ and $v^1_{\sizew}\in \CCS(\Bim{k}(u))$, then $v^1_{\sizew}\leq d(u)\dotdiv 1$ and $d(\Bim{k+1}(v^2_{\sizew}))\leq d(u)\dotdiv 2$.
Consequently, $d(v_{i-1}^2\setminus v_{\sizew}^1) \leq d(u)\dotdiv 1$. In the latter case, $v_{i-1}^2\in \CCS(\Bim{k+1}(v_{i}^2)\cup v_{i}^1)$; and since $v_{i}^2=v_{i}^2\cup v_{i+1}^1=(v_{i}^2\setminus v_{i+1}^1)\cup v_{i+1}^1$, and $\Bim{k+1}(A\cup B)=\Bim{k+1}(A)\cup \Bim{k+1}(B)$, we have $v_{i-1}^2\in \CCS(\Bim{k+1}(v_{i}^2\setminus v_{i+1}^1)\cup \Bim{k+1}(v_{i+1}^1)\cup v_{i}^1)$; but $\Bim{k+1}(v_{i+1}^1)\subseteq v_{i}^1$, hence $v_{i-1}^2\in \CCS(\Bim{k+1}(v_{i}^2\setminus v_{i+1}^1)\cup v_{i}^1)$. 
By Prop.\ \ref{prop-CCS}.\ref{Three}: $\exists v'\colon v'\in\CCS(\Bim{k+1}(v_{i}^2\setminus v_{i+1}^1))$ and $v_{i-1}^2=v_{i}^1\cup v'$. Thus $v_{i-1}^2\setminus v_{i}^1\subseteq v'$, and  $d(v_{i-1}^2\setminus v_{i}^1)\leq d(v')=d(\Bim{k+1}(v_{i}^2\setminus v_{i+1}^1))\leq d(v_{i}^2\setminus v_{i+1}^1)\dotdiv 1 \leq d(u)+(i+1)\dotdiv (\sizew+1)\dotdiv 1$ (by induction hypothesis) $\leq d(u)+i+\dotdiv (\sizew+1)$.  \end{proposition}

Now, we check that $\longwin{z}{D}{l+1}{l}$ is indeed a $(k,\sizew+1,\lambda)$-window for $(u,v_0)$ by examining the definition of continuations. \\
Since on the one hand $z_0,z_1,z_2,\cdots,z_{\sizew+1}=v^1_0,v^2_0,v^2_1,\cdots,v^2_{\sizew}$, and on the other hand $D_0,D_1,D_2,\cdots,D_{\sizew}=W^1_0,W^2_0,W^2_1,\cdots,W^2_{\sizew-1}$, we have:
\begin{enumerate}
    \item $v_{\sizew}^2 \in \CCS(\Bim{k}(u))$
    \item 
    \begin{enumerate}
    \item $v_{\sizew-1}^2\in\CCS(v_{\sizew}^1\cup\Bim{k+1}(v_{\sizew}^2))$, and since $v_{\sizew}^1\in\CCS(\Bim{k}(u))$, it comes $v_{\sizew-1}^2\in\CCS(\Bim{k}(u)\cup\Bim{k+1}(v_{\sizew}^2))$;
    \item take $i\in[0:\sizew-2]$. Then $v_{i-1}^2\in\CCS(v_i^1\cup\Bim{k+1}(v_i^2))$, and since $v_i^1\in\CCS(\Bim{k}(u)\cup\Bim{k+1}(v_{i+1}^1)$, by Prop.\ \ref{prop-CCS}.\ref{One}: $v_{i-1}^2\in\CCS(\Bim{k}(u)\cup\Bim{k+1}(v_i^2)\cup\Bim{k+1}(v_{i+1}^1))$. But $v_{i+1}^1\inc v_i^2$, hence $v_{i-1}^2\in\CCS(\Bim{k}(u)\cup\Bim{k+1}(v_i^2))$; 
    \item to conclude for condition 2, it remains to prove that $v_0^1\in \CCS(\Bim{k+1}(v_0^2)\cup \Bim{k}(u))$. By Prop.\ \ref{proposition-degree}, $d(v_0^2\setminus v_1^1)\leq d(u)\dotdiv\sizew\leq 0$, hence if $\Bi{k+1} \phi\in v_0^2$ then $\Bi{b} \phi\in v_1^1$ and thus $\Bim{k+1}(v_0^2)=\Bim{k+1}(v_1^1)$. Since $v_0^1\in \CCS(\Bim{k+1}(v_1^1)\cup \Bim{k}(u))$, then $v_0^1\in \CCS(\Bim{k+1}(v_0^2)\cup \Bim{k}(u))$.
    \end{enumerate}
    \item We verify condition 3 by proving that for each $i\in\intero{0}{\sizew}$ $W^2_i$ is a $(k+1,\lambda(v^2_{i+1}),\lambda)$-window for $(v^2_{i+1},v^2_i)$ and that $W^1_0$ is a $(k+1,\lambda(v^2_0),\lambda)$-window for $(v^2_0,v^1_0)$. It is immediate for the first ones. 
    Concerning $W^1_0$: if $k=\pi-1$ and $W^1_0=\voidpair$ then we are done, else let $W^1_0=\langle(v^{0,1}_i),(W^{0,1}_i)\rangle$ (we omit the ranges of the sequences) ; since it is a $(k+1,\lambda(v^1_1),\lambda)$-window for $(v^1_1,v^1_0)$ we have, for all $0\leq i< \lambda(v^1_1)\colon v^{0,1}_i\in\CCS(\Bim{k+1}(v^{0,1}_{i+1})\cup\Bim{k}(v^1_1))$, but recall from 3.c above that $\Bim{k}(v^2_0)\inc\Bim{k}(v^1_1)$, and since each $W^{0,1}_i$ is a $(k+2,\lambda(v^{0,1}_{i+1}),\lambda)$-window for $(v^{0,1}_{i+1},v^{0,1}_i)$, hence conditions are met to state that $W^1_0$ is a $(k+1,\lambda(v^2_0),\lambda)$-window for $(v^2_0,v^1_0)$. 
    \end{enumerate}
    
\end{proof}

\begin{lemma}\label{prolongation-to-infinite-window}
    Let $u$, $v_0$ be two \CCS, let $k\in \Pi$, if there exists a $(k,\chi(u),d(u))$-window for $(u,v_0)$ for a ``sufficiently large'' $\chi(u)$ which depends on $u$, then there exists a $(k,\infty,d)$-window for $(u,v_0)$. Such a window will be called ``maximal''. 
\end{lemma}
\begin{proof}
    In order to precise ``sufficiently large", let us first compute the number of $\CCS$ in a $(k,d(u),d)$-window which is either $W=\win{v}{W}{d(u)}$ or $W=\voidpair$ according to $k$. If we let $n=\pi-k$, this number is bounded above by the following recurrent inequalities: \\
    $$\begin{array}{ll}
         s(n) &  = (\mbox{if } n=0 \mbox{ then } 1 \mbox{ else }{d(u)}+\Sigma_{i=0}^{d(u)-1}s(n-1)\\
         & \leq (\mbox{if } n=0 \mbox{ then } 1 \mbox{ else } d(u) + d(u).s(n-1)\\
         &\leq d(u) + d^2(u)+ d^2.s(n-1)\\
         &\leq d(u) + d^2(u)+ \cdot +d^n(u)+ d^n.s(0)\\
         &\leq Q(d(u)) \hfill\mbox{\indent for some polynomial $Q$ of degree $\pi+1$}\\
         &\leq Q(\lgu)\\
    \end{array}$$

    Each \CCS\ is a member of $\SF(u)$ and there are at most $2^{\lgx{\SF(u)}}=2^{c_{\mathrm{ccs}}.\lgu}$ of them. Hence, there are at most $(2^{c_{\mathrm{ccs}}.\lgu})^{Q(\lgu)}$ distinct $(k,d(u),d))$-windows for $(u,v_0)$, i.e.\ $2^{P(\lgu)}$ for some polynomial $P$ of degree $\pi+2$. \\
We claim that $\chi(u)=2^{P(\lgu)}+d(u)$: let $W$ be a $(k,\uplim{u},d(u))$-window for $(u,v_0)$, it can be broken into the sequence $(W_j)_{j\in[0\:\uplim{u}]}$ of $(k,d(u),d))$-partial windows for $(u,v_0)$ each of them being a $k$-continuation of the previous. Then, because of the above bound, at least two of them are identical: there exists integers $h,\delta$ such that $\delta\neq 0$ and $h+\delta \leq \uplim{u}$ and $W_h=W_{h+\delta}$. Let $(\Wpp_j)_{0\leq j}$ be the infinite sequence such that for all $j\leq h$, $\Wpp_j=W_j$ and for all $j>h$, $\Wpp_j=W_{h+((j-h)\!\!\!\mod \delta)}$. By construction, for all $j\geq 0$, $\Wpp_{j+1}$ is a $k$-continuation of $\Wpp_j$. For all $j\geq 0$, suppose that $\Wpp_j=\win{v^j}{W^j}{d(u)}$, and set $W=\langle (v^i_0)_{0\leq i}, (W^i_0)_{0\leq i}\rangle$ which is a $(k,\infty,d)$-window for $(u,v_0)$. 
    \end{proof}

\begin{lemma}\label{from-model-to-window}
    Let $M=(S,(\Ri{i})_{i\in\Pi},V)$ be a $\Pi$-dense model. Let $u,v_0$ be two \CCS\ and suppose that there exists $x,y_0\in S$ such that: $(x,y_0)\in\Ri{k}$ and $M,x\models u$ and $M,y_0\models v_0$. Then for any integer $\sizew$, there exists $W$ a $(k,\chi(u),\chi)$-window and all its $\CCS$ are $\Kpi$-satisfiable (i.e.\ for $(u,v_0)$ such that for all $v\in W\colon v$ is $\Kpi$-satisfiable). 
\end{lemma}
\begin{proof}[By descending induction on $k$]\\

\begin{itemize}
    \item If $k=\pi$, we are done with $\voidpair$. 
    \item If $k<\pi$, since $(x,y_0)\in\Ri{k}$ and $M$ is $\Pi$-dense, let $(y_i)_{i\geq 0}$ be the $(k,\infty)$-sequence for $(x,y_0)$ in $(S,(\Ri{i})_{i\in\Pi})$.\\
    For each $i\geq 1$ let $v_i=\SF(\Bim{k}(u))\cup y_i$. Trivially $v_i$ is $\Kpi$-satisfiable. \\
    First, we establish the following fact: for all $i\geq 0$, $v_i\in\CCS(\Bim{k}(u)\cup\Bim{k+1}(v_{i+1}))$: 
    \begin{enumerate}
        \item since $M,x\models u$ and $(x,y_i)\in\Ri{k}$, then $M,y_i\models \Bim{k}(u)$, hence $\Bim{k}(u)\inc \SF(\Bim{k}(u))\cap y_i$, thus $\Bim{k}(u)\inc v_i$, 2) let $\phi\in\Bim{k+1}(v_{i+1})$, then $\Bi{k+1}\phi\in v_{i+1}$, hence $\Bi{k+1}\phi\in\SF(\Bim{k}(u))\cap y_{i+1}$, thus $\phi\in\SF(\Bim{k}(u))\cap y_{i}$, i.e.\ $\phi\in v_i$; finally, $\Bim{k}(u)\cup\Bim{k+1}(v_{i+1})\inc v_i$; 
        \item $v_i$ is saturated; we only consider the case of the $\wedge$: let $(\phi\wedge\psi)\in s_i$, hence $(\phi\wedge\psi)\in\SF(\Bim{k}(u))\cap y_i$, hence $\phi\in\SF(\Bim{k}(u))\cap y_i$ (as well for $\psi$), hence $\phi\in y_i$ (idem for $\psi$); 
        \item $v_i$ is consistent, since it is a finite subset of the consistent set $y_i$. 
    \end{enumerate}
    1.\ to 3.\ together prove the fact. \\
    Thus, for each $i\geq 0$, both $v_{i+1}$ and $v_i$ are \CCS\ such that there exists $y_{i+1}$ and $y_i$ with $(y_{i+1},y_i)\in\Ri{k+1}$ and $M,y_{i+1}\models v_{i+1}$ and $M,y_i\models v_i$, hence induction hypothesis applies: there exists a $(k+1,\uplim{v_{i+1}},\chi)$-window for $(v_{i+1},v_i)$ with all its $\CCS$ $\Kpi$-satisfiable. Let us denote it $W_i$. Finally, and since by hypothesis $v_0$ is $\Kpi$-satisfiable, $\win{v}{W}{\sizew}$ is the desired $(k,\uplim{u},\chi)$-window for $(u,v_0)$. 
\end{itemize}
\end{proof}

\begin{corollary}\label{existence-of-window}
Let $u$ a \CCS\ containing some formula $\neg \Bi{k}\phi$ and satisfied at a world $x$ of a $\Pi$-dense model $M$ then a) there exists $v_0\in\CCS(\{\neg\phi\}\cup\Bim{k}(u))$ and b) there exists a $(k,\uplim{u},\chi)$-window for $(u,v_0)$ and all its $\CCS$ are $\Kpi$-satisfiable.
\end{corollary}
\begin{proof}
    Since $M,x\models u$ then $M,x\models \neg\Bi{k}\phi\wedge \Bi{k}(u)$, hence there exists $y_0$: $(x,y_0)\in\Ri{k}$ and $M,y_0\models \neg\phi\wedge \Bim{k}(u)$. Let $v_0=\SF(u)\cap y_0$ and conclude with the above lemma with $n=d(u)$. 
\end{proof}
\section{The algorithm}\label{algorithm}

The idea is that despite the infinity of $\Pi$-dense models, and because of lemma \ref{prolongation-to-infinite-window}, it would suffices to check $(k,\uplim{u},\chi)$-windows. But they are of exponential size, so we need to check them by exploring small (=polynomial) pieces at a time and this will appear to be recursively possible, thanks to continuations. 

We provide intuition in Fig.\ \ref{fig:recursive:windows} with $d(u)=1$ and $\Pi=\{0\cdots 3\}$ (surely $d(u)$ is ridiculously small, but this is for keeping the window small enough). Above is a window, and below a continuation of it (pointwise included by $\sqsubseteq$). As one can ``see'', boxed formulas between $\tilde v_2$ and $\tilde v_1$ cannot influence those between $v_0$ and $v_1$ (since they are of degree 0). Hence, provided the subwindow between the latter two is satisfiable, we can forget it and proceed to try to extend the window. For this, we have to test for $\sqsubseteq$-inclusion. Of course the same reasoning applies at each scale on subwindows. \\

\begin{figure}[!h]
\begin{centering}
  \fbox{
  \begin{tikzpicture}[scale=0.5]
\coordinate (u) at (4,8) ;
\coordinate (v0) at (16,8) ;
\coordinate (v1) at (12,7) ;

\coordinate (v2) at (8,6) ;
\coordinate (v0p) at (12,3) ;
\coordinate (v1p) at (8,2) ;
\coordinate (v2p) at (4,1) ;
\coordinate (v2b) at (4,-1) ;
\coordinate (v1b) at (8,0) ;
\coordinate (arr) at (10,3.5) ;

\coordinate (v02) at (12,6);
\coordinate (v01) at (14,7) ;
\coordinate (v12) at (8,5) ;
\coordinate (v11) at (10,6) ;
\coordinate (v001) at (15,7) ;
\coordinate (v002) at (14,6) ;
\coordinate (v012) at (12,5) ;
\coordinate (v011) at (13,6) ;
\coordinate (v101) at (11,6) ;
\coordinate (v102) at (10,5) ;
\coordinate (v111) at (9,5) ;
\coordinate (v112) at (8,4) ;
\draw (u) node{$\bullet$} node[above]{$u$};
\draw (v0) node{$\bullet$} node[above]{$v_0$};
\draw (v1) node{$\bullet$} node[above]{$v_1$};
\draw (v2) node{$\bullet$} node[above]{$v_2$};
\draw (v01) node{$\bullet$};
\draw (v02) node{$\bullet$};
\draw (v11) node{$\bullet$};
\draw (v12) node{$\bullet$};
\draw (arr) node{\rotatebox[origin=c]{-90}{$\sqsubseteq$}};

\draw (v001) node{$\bullet$};
\draw (v002) node{$\bullet$};
\draw (v011) node{$\bullet$};
\draw (v012) node{$\bullet$};
\draw (v101) node{$\bullet$};
\draw (v102) node{$\bullet$};
\draw (v111) node{$\bullet$};
\draw (v112) node{$\bullet$};
\draw (v0p) node{$\bullet$} node[above]{$\tilde v_0$};
\draw (v1p) node{$\bullet$} node[above]{$\tilde v_1$};
\draw (v2p) node{$\bullet$} node[above]{$\tilde v_2$};

\draw [dashed] (u) -- (v0);
\draw [dashed] (u) -- (v1);
\draw [dashed] (u) -- (v2);
\draw (v2) -- (v1) -- (v0);
\draw (v2) -- (v12);
\draw (v2) -- (v11);
\draw [densely dotted] (v12) -- (v112);
\draw [densely dotted] (v12) -- (v111);
\draw [densely dotted] (v12) -- (v11) -- (v1);
\draw [loosely dotted] (v112) -- (v111) -- (v11);
\draw [densely dotted] (v11) -- (v102);
\draw [densely dotted] (v11) -- (v101);
\draw [loosely dotted] (v102) -- (v101) -- (v1);
\draw (v2) -- (v1) -- (v0);
\draw (v1) -- (v02);
\draw (v1) -- (v01);
\draw [densely dotted] (v02) -- (v01) -- (v0);
\draw [densely dotted] (v01) -- (v002);
\draw [densely dotted] (v01) -- (v001);
\draw [densely dotted] (v02) -- (v012);
\draw [densely dotted] (v02) -- (v011);
\draw [loosely dotted] (v012) -- (v011) -- (v01);
\draw [loosely dotted] (v002) -- (v001) -- (v0);

\draw (v0p) -- (v1p) -- (v1b) -- cycle;
\draw (v1p) -- (v2p) -- (v2b) -- cycle;
\end{tikzpicture}
}
\caption{A $2$-window and one of its potential continuation (arrows are left-to-right or else top-bottom, style is specific to each $i\in\Pi$)\\}\label{fig:recursive:windows}
\end{centering}
\end{figure}
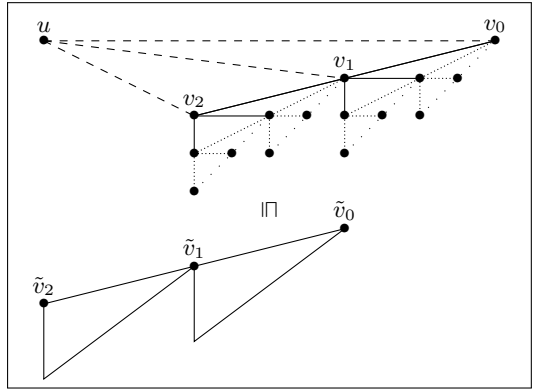
The algorithm we present below is based on the function $\sat$ which answers to the $\Kpi$-satisfiability of its argument. 
Because of Prop. \ref{prop-CCS}.\ref{Five}, the $\Kpi$-satisfiability of a set $s$ of formulas amounts to that of at least one of its $\CCS$, since $s$ is $\Kpi$-satisfiable if and only if there exists a $\Kpi$-satisfiable $u\in\CCS(s)$. Hence, given an initial set of formulas $s$ to be tested, the initial call is $\sat(\chooseCCS(\{s\}))$.\\
In what follows we use built-in functions \algand\ and \all.
The former function lazily implements a logical ``and".
The latter function lazily tests if all members of its list argument are true. \\
\setlength{\textfloatsep}{0pt}
\setlength{\floatsep}{0pt}
\begin{algorithm}[H]
 \floatname{algorithm}{Function}
\begin{algorithmic}
\caption{Test for \Kpi-satisfiability of a \CCS: must be classically consistent and recursively tests each $\Diamond$-formula subsequent window.}
\Function{\sat}{$u$}:
\State {return}
\State {\hspace{0.3cm}$u\neq \{\bot\}$}
\State {\hspace{0.1cm}\algand\ }
\State {\hspace{0.2cm}$\all\{\sat(\chooseCCS(\{\neg\phi\}\cup \Bim{\pi}(u'))\colon \neg\Bi{\pi}\phi\in u'\}$}
\State {\hspace{0.1cm}\algand\ }
\State {\hspace{0.2cm}$\all \{\satW(\chooseW(u,\chooseCCS(\{\neg\phi\}\cup\Bim{k}(u)),k), u,k,\uplim{u})\colon\neg\Bi{k}\phi\in u\}$}
\EndFunction
\end{algorithmic}
\end{algorithm}
\vspace{-2cm}

\begin{algorithm}[H]
\floatname{algorithm}{Function}
\begin{algorithmic}
\caption{Returns $\{\bot\}$ if $s$ is not classically consistent, otherwise returns one classically saturated open branch non-deterministically chosen}
\Function{\chooseCCS}{$s$}
\If  {$\CCS(s)\neq \emptyset$}
\State {return one $u\in \CCS(s)$}
\Else 
\State {return $\{\bot\}$}
\EndIf
\EndFunction
\end{algorithmic}
\end{algorithm}
\vspace{-3cm}

\begin{algorithm}[H]
\floatname{algorithm}{Function}
\begin{algorithmic}
\caption{Non-deterministically chooses a $(k,d(u),d)$-window for $(u,v)$}
\Function{\chooseW}{$u$,$v$,$k$}
\If {there exists a $(k,d(u),d)$-window $W$ for $(u,v)$}
\State {return $W$}
\Else 
\State {return $\langle (\{\bot\})_{i\in[0:d(u)]}, \emptyset\rangle $}
\EndIf
\EndFunction
\end{algorithmic}
\end{algorithm}

\begin{algorithm}[H]
\floatname{algorithm}{Function}
\begin{algorithmic}
\caption{Tests the satisfiability of a $(k,d(u),d)$-window for $(u,v_0)$ and recursively that of each of its subwindows and continuations, until a repetition happens or a contradiction is detected}
\Function{\satW}{$W$,$u$,$k$,$N$}:\algorithmiccomment{$W$ is $\win{v}{W}{d(u)}$
 if $k<\pi$}
\If {$N=0$ \textbf{or} $k=\pi$} \algorithmiccomment{or $\voidpair$ if $k=\pi$}
\State {return \true}   \algorithmiccomment{or $\langle (\{\bot\})_{i\in[0:d(u)]}, \emptyset\rangle$ if $\chooseW$ or $\nextW$ has failed}
\Else 
\State {return}
\State {\hspace{0.83cm} \sat$(v_0)$}
\State {\hspace{0.2cm} $\algand\ \satW(W_0,v_{1},k+1,\uplim{v_1})$}\
\State {\hspace{0.2cm} \algand\ \satW(\nextW$(W,u,k),u,k,N-1))$}
\EndIf
\EndFunction
\end{algorithmic}
\end{algorithm}

\begin{algorithm}[H]
\floatname{algorithm}{Function}
\begin{algorithmic}
\caption{Non-deterministically chooses a $k$-continuation of a window for $s$}
\Function{\nextW}{$W$,$u$}
\If {there exists a $k$-continuation $W_1$ of $W$ for $u$}
\State {return $W_1$}
\Else 
\State {return $\langle (\{\bot\})_{0\leq i\leq d(u)}, \emptyset\rangle $}
\EndIf
\EndFunction
\end{algorithmic}
\end{algorithm}

\section{Analysis of the algorithm}\label{section:complexity:of:Kpi}

Given a \Kpi-model $M=(S,\Ri{a},\Ri{b},v)$ and a set $s$ of formulas, we will write $M,x\models s$ for $\forall \phi\in s\colon M,x\models \phi$. 

\begin{proposition}\label{proposition-on-members}
\mbox{}\\
\vspace{-0.5cm}
\begin{enumerate}
    \item For all $j\in\interf{0}{\uplim{u}}$ we have $\{W\lbr  j\rbr \}\subseteq \{W\}$. 
    \item Given an initial call $\satW(W'\lbr  0\rbr ,u',k',\chi(u'))$, then in all subsequent calls $\satW(W\lbr  j\rbr ,u,k,N)$ the precondition $\{W\}\subseteq \{W'\}$ is satisfied. 
\end{enumerate}
\end{proposition}
\begin{proof} 
\begin{enumerate}
    \item
    If $W=\voidpair$, then we are done since $\{\voidpair\lbr j\rbr \}=\emptyset\subseteq\{W\}$\\
    else, suppose $W\lbr j\rbr =\shrinkwin{j}{v}{W}{j+d(u)}$; then $\{W\lbr j\rbr \}=\{v_i\colon i\in\interf{j}{j+d(u)}\cup\bigcup_{i\in[j:j+d(u)[}\{W_i\lbr 0\rbr \}$. Since, a) $\{v_i\colon i\in\interf{j}{j+d(u)}\}\subseteq \{v_i\colon i\in\interf{0}{\uplim{u}}\}\subseteq \{W\}$ and b) by IH $\{W_i\lbr 0\rbr \}\subseteq \{W_i\}$ and $\{W_i\}\subseteq \{W\}$. \\
    \item Initially, it is true for $\satW(W'\lbr 0\rbr )$ by 1) above. It remains true for the subsequent calls $\satW(W_j\lbr 0\rbr ,\cdots)$ since by 1)  $\{W_j\lbr 0\rbr \}\subseteq \{W_j\}$ and since $\{W_j\}\subseteq \{W\}$, we have by IH $\{W\}\subseteq \{W'\}$. It also remains true for the calls $\satW(W\lbr j+1\rbr ,\cdots)$ since as a partial window $\{W\lbr j+1\rbr \}\subseteq \{W\}$ and we conclude again by IH. 
    \end{enumerate}
\end{proof}

\begin{lemma}[Soundness]\label{soundness}\\
If $u'$ is a \Kpi-satisfiable \CCS\ then the call $\sat(u')$ returns \true. 
\end{lemma}

\begin{proof} 
Since $u'$ is \Kpi-satisfiable, then $u'\neq \{\bot\}$. Hence the result of $\sat(u')$ rely on that of: 
\begin{itemize}
    \item [] $\all \{\sat(\chooseCCS(\{\neg\phi\}\cup \Bim{\pi}(u'))\colon \neg\Bi{\pi}\phi\in u'\}$
    \item [] $\algand $
    \item [] $\all \{\satW(\chooseW(u',\chooseCCS(\{\neg\phi\}\cup\Bim{k}(u')),k'), u',k',\uplim{u'})\colon$ 
    \item []\hspace{2cm} $k'\in\Pi^-\mbox{ and }\neg\Bi{k}\phi\in u'\}$
\end{itemize}
We proceed by induction on $d(u')$\\
1) Case $d(u')=0$: then the sets\\
    $\{\sat(\chooseCCS(\{\neg\phi\}\cup \Bim{\pi}(u'))\colon \neg\Bi{\pi}\phi\in u'\}$ and \\
    $\{\satW(\chooseW(u',\chooseCCS(\{\neg\phi\}\cup\Bim{k'}(u')),k'), u',k',\uplim{u'})\colon\neg\Bi{k}\phi\in u'\}$\\
    are empty. Hence $\sat(u')$ returns \true.\\
2) Case $d(u')\geq 1$. The induction hypothesis is IH$_1$: if $u$ is $\Kpi$-satisfiable and $d(u)<d(u')$ then $\satW(u)$ returns \true. Now, for each $\neg\Bi{k'}\phi\in u'$:\\
    2.1) if $k=\pi$ then by Corollary \ref{existence-of-window}, there exists $v'_0\in\CCS(\{\neg\phi\}\cup\Bim{\pi}(u'))$ and there exists $W'$ a $(\pi,\chi(u'),\chi)$-window for $(u',v'_0)$ with all its $\CCS$ $\Kpi$-satisfiable, namely $W'=\voidpair$.
    Thus by IH$_1$ (since $d(v'_0)<d(u')$), there exists $v'_0 \in \CCS(\{\neg\phi\}\cup \Bim{\pi}(u'))$ such that $\sat(v'_0)$ returns \true.\\
    Hence $\{\sat(\chooseCCS(\{\neg\phi\}\cup \Bim{\pi}(u'))\colon \neg\Bi{\pi}\phi\in w\}$ returns \true. \\
    2.2) if $k\in\Pi^-$. By Corollary \ref{existence-of-window}, there exists $v'_0\in\CCS(\{\neg\phi\}\cup\Bim{k'}(u'))$ and there exists $W'$ a $(k',\chi(u'),\chi)$-window for $(u',v'_0)$ with all its $\CCS$ $\Kpi$-satisfiable. We set:
    \begin{itemize}
        \item $\chooseCCS(\{\neg\phi\}\cup\Bim{k}(u'))=v_0$
        \item $\chooseW(u',v'_0,k)=W'\lbr 0\rbr $
        \item and for each subsequent call $\satW(W\lbr 0\rbr ,u,k,\uplim{u})$ \\
        let $\nextW(W\lbr j\rbr ,u,k)=W\lbr j+1\rbr $ for $j\in\intero{0}{\chi(u)}$ \hfill(its $k$-continuation)
    \end{itemize}
Given that the initial call $\satW(W'\lbr 0\rbr ,u',k',\chi(u'))$, and that all \CCS\ of $W'$ are $\Kpi$-satisfiable, then all subsequent calls $\satW(W\lbr j\rbr ,u,k,N)$ return \true. This can be proved by the following nested induction on $(\pi-k,N)$:\\
    \begin{itemize}
        \item if $k=\pi$ (and $W=\voidpair$) or $N=0$ it is true since $\satW(W,\cdots)=\true$. 
        \item else ($k<\pi$ and $N>0$) with IH$_2$: $\satW(W\lbr j\rbr ,u,k,N)$ return \true; \\
        then: 
    \end{itemize}
        \begin{equation*}
        \begin{array}{lll}
            \satW(W\lbr j\rbr ,u,k,N) &=&  \sat(v_0) \\ 
                              & &  \algand\,\,\satW(W_j\lbr 0\rbr ,v_{j+1},k+1,N)\\
                              & & \algand\,\,\satW(W\lbr j+1\rbr ,u,k,N-1)\\
                              &=&  \true \mbox{ (since by Prop. \ref{proposition-on-members}, $v_0\in \{W'\}$ and, as such,}\\
                              & &  \mbox{is satisfiable, thus $\sat(v_0)$ returns $\true$ by IH$_1$)}\\ 
                              & &  \algand\,\,\true \mbox{(by IH$_2$ since $(\pi-(k+1)+N<\pi-k+N$)}\\
                              & & \algand\,\,\true \mbox{(by IH$_2$ since $\pi-k+N-1<\pi-k+N$)} 
        \end{array} 
        \end{equation*}
        In particular, $\satW(W'\lbr 0\rbr ,u',k',\uplim{u'})$ returns \true, and so does \\
     $\satW(\chooseW(u',\chooseCCS(\{\neg\phi\}\cup\Bim{k'}(u')),k'), u',k',\uplim{u'})$. \\ Consequently, $\sat(u')$ returns \true\ too.
  \end{proof}

For proving the completeness of this algorithm, we need to transform a \true\ into a model. To this aim we define the notion of satisfiability of a window. 
\paragraph{Satisfiability of window}
Let $M=(S,(\Ri{i})_{i\in\Pi},V)$ be a $\Pi$-dense model and $x\in S$. Let $W$ be a $(k,n,\lambda)$-window for $(u,v_0)$. We say that $M$ satisfies $W$ at $x$, denoted by $M,x\models W$ iff:
\begin{itemize}
    \item $W=\voidpair$
    \item or, if $W=\win{v}{W}{n}$
    \begin{itemize}
        \item $M,x\models u$
        \item $\exists y_0\in\Ri{k}(x)\colon M,y_0\models v_0$
        \item if $i\in\intero{0}{n}$: 
        $\exists y_{i+1}\in\Ri{k+1}^-(y_{i})\cap \Ri{k}(u)\colon M,y_{i+1}\models v_{i+1}$ and $M,y_{i+1}\models W_{i}$
    \end{itemize}
\end{itemize}

\begin{lemma}[Completeness]\label{completeness}\\
Given a \CCS\ $u$ and a $(k,d(u),d)$-window for $(u,v_0)$, then:
\begin{itemize}
    \item [$\bullet$] if $\sat(u)$ returns \true\ then $u$ is $\Kpi$-satisfiable
    \item [$\bullet$] if $d(u)\neq 0$ and $\satW(W,u,k,N)$ returns \true\ then $W$ is \Kpi-satisfiable
\end{itemize}
\end{lemma}
\begin{proof}
    We construct $M=(S,(\Ri{i})_{i\in\Pi},V)$ by induction on $d(u)$. For each $k\in\Pi$ let $\neg\Bi{k}\phi^k_0,\cdots,\neg\Bi{k}\phi^k_{n_k}$ be the $\Diamond$-formulas of $u$. In what follows we define $V_x$ by: $V_x(p)=\{x\}$ if $p\in x$ and else $V_x(p)=\emptyset$, for all $p\in\At$. \\
    If $d(u)=0$ the model $M=(\{u\}, \voidsequence,V_u)$ $u$. Else, 
    \begin{itemize}
        \item Firstly, for each $\neg\Bi{\pi}\phi_l^{\pi}\in u$ suppose that
        $\sat(\chooseCCS(\{\neg\phi_l^{\pi}\}\cup \Bim{\pi}(u)))$ returns \true. By induction hypothesis, $\{\neg\phi^{\pi}_l\}\cup \Bim{\pi}(u)$ is true in some \Kpi-model: $(S^{\pi,l},(R_i^{\pi,l})_{i\in\Pi},V^{\pi,l}),v^{\pi,l}\models \{\neg\phi^{\pi}_l\}\cup \Bim{\pi}(u)$
        \item Secondly, for each $k\in\Pi^-$ and for each $\neg\Bim{k}\phi^k_l\in s$ suppose that the call $\satW(\chooseW(u,\chooseCCS(\{\neg\phi_l\}\cup\Bim{k}(u)),k), u,k,\uplim{u})$ returns \true. Let $v_0$ be the \CCS\ chosen by $\chooseCCS$, let $W^0=\win{v^0}{W}{u}$ the $(k,d(u),d)$-window for $(u,v_0)$ chosen by $\chooseW$, and for $i\in\interf{1}{\uplim{u}}$, let $W^{j}=\win{v^j}{W^j}{u}$ be the $(k,d(u),d)$-windows for $(u,v_i)$ chosen by the successive recursive calls to $\nextW$ (which succeed by hypothesis). Since each $W^{j+1}$ is a $k$-continuation of $W^j$, by repeated application of lemma \ref{k2kplusone-window}, with $(v_i,W_i)=(v^i_0,W^i_0)$ we obtain the following $(k,\uplim{u},d)$-window for $(u,v_0)$: $\win{v}{W}{\uplim{u}}$. Now, by applying lemma \ref{prolongation-to-infinite-window}, we can extend it to a $(k,\infty,d)$-window for $(u,v_0)$: 
        $W=\infwin{v}{W}$ where, beyond $\uplim{u}$, all $v_i$ and $W_i$ are copies of a $v_j$ and a $W_j$ with $j\leq \uplim{u}$. Since by hypothesis, for $i\geq 0$ all calls $\sat(v_i)$ and $\satW(W_i,v_{i+1},k+1,\chi(v_{i+1}))$ returns \true, then by induction hypothesis $v_i$ and $W_i$ are $\Kpi$-satisfiable, let $M^0,x^0\models v_0$ and $M^{i+1},x^{i+1}\models W_i$ (this implies $M^{i+1},x^{i+1}\models v_{i+1}$ by definition).  \\
        In fact since they all depend on the formula $\neg\Bi{k}\phi^k_l$ involved, we add $k,l$ in the superscript giving: $W^{k,l}$, $v^{k,l}_i$, $W^{k,l}_i$, instead of just $W$, $v_i$ and $W_i$ and we write $M^{k,l,i},y^{k,l,i}\models W^{k,l}_{i-1}$, and $M^{k,l,i},y^{k,l,i}\models v_i^{k,l}$,  with $M^{k,l,i}=(S^{k,l,i},(R^{k,l,i}_j)_{j\in\Pi},V^{k,l,i})$. We merge these models into one, for each $\neg\Bi{k}\phi^k_l$ formula of $u$: $M^{k,l}=\bigsqcup_{i\geq 0} (M^{k,l,i})$. \\
        Putting all things together, we define:
        \[M'=(S',(\Ri{k}')_{k\in\Pi},V')=\bigsqcup_{k\in\Pi,l\in[1:n_k],\neg\Bi{k}\phi^k_l} M^{k,l}
        \]
        \end{itemize}
        $M'$ is a $\Pi$-dense model since it is the disjoint union of $\Pi$-dense models. It remains to connect it with $u$ seen as a possible world to form the final model $M=(S,(\Ri{k})_{k\in\Pi},V)$: 
        \begin{itemize}
            \item $S=S'\sqcup \{u\}$
            \item for each $k\in\Pi$: $\Ri{k}''=\Ri{k}'\bigsqcup_{l\in[0:n_k],\neg\Bi{k}\phi^k_l,i\geq 0}\{(u,y^{k,l,i})\}$
            \item for each $k\in\Pi^-$: $\Ri{k+1}=\Ri{k+1}''\sqcup\bigsqcup_{l\in[0:n_k],\neg\Bi{k}\phi^k_l,i\geq 0}\{(y^{k+1,l,i+1},y^{k+1,l,i})\}$
            \item for each $p\in\At$: $V(p)=V'(p)\sqcup V_{u}(p)$
            \end{itemize}
        Now, it is time to check that 1. $M$ is $\Pi$-dense, and (truth lemma) both 2. $M,u\models u$ and 3. $M,u\models  W^{k,l}$ are true:
    \begin{enumerate}
        \item Let $(x,y)\in\Ri{k}$ for some $k\in\Pi^-$: 
        \begin{itemize}
            \item if $(x,y)\in\Ri{k}''$, i.e.\ $(x,y)=(u,y^{k,l,i})$ then since $(u,y^{k,l,i+1})\in\Ri{k}$ and $(y^{k,l,i+1},y^{k,l,i})\in\Ri{k+1}$, we are done;
            \item if $(x,y)\in\Ri{k}'$ and $\not\in\Ri{k}''$, i.e.\ $(x,y)=(y^{k+1,l,i+1},y^{k+1,l,i})$, then:
            \begin{itemize}
                \item if $k=\pi-1$ then we are done;
                \item else since $W_i^{k,l}$ is $\Kpi$-satisfiable, $\exists z\colon (y^{k+1,l,i+1},z)\in\Ri{k+2}$ and $(y^{k+1,l,i},z)\in\Ri{k+1}$.
            \end{itemize}
        \end{itemize}
        \item We only treat the case of modal formulas. \\
        $\Diamond$-formulas: Let $\neg\Bi{k}\neg\phi^k_l\in u$, since $\neg\phi^k_l\in v^{k,l}_0$ and $M,y^{k,l,0}\models v^{k,l}_0$, and  $(u,y^{k,l,0})\in\Ri{k}$, we are done;\\
        $\Box$-formulas: Let $\Bi{k}\neg\phi^k_l\in u$ and let $(u,y^{k,l,i})\in\Ri{k}$ for some $i$, $\neg\phi^k_l\in\Bim{k}(u)$ hence $\neg\phi^k_l\in v_i^{k,l}$ which is a member of $W^{k,l}$, the $(k,\infty,\infty)$-window for $(u,v_0^{k,l})$ defined above. But since we added edges between worlds $y_i^{k,l}$ of distinct models, we must check that still $M,y_i^{k,l}\models v_i^{k,l}$ since $v_i^{k,l}$ may contain a $\Box$-formula that would be unsatisfied in $M$. We do so by a short induction on $i$; this is true for $M,y_{0}^{k,l}\models v_0^{k,l}$ (since no edge from $v_0^{k,l}$ were added); by definition of windows, we have that $\Bim{k+1} v_{i+1}^{k,l}\subseteq v_i^{k,l}$ and since $M,y_i^{k,l}\models v_i^{k,l}$ (by IH on $i$) hence $M,y_{i+1}^{k,l}\models v_{i+1}^{k,l}$ and we are done again. 
        \item 
        \begin{itemize}
            \item [a)] $M,u\models u$ 
        \item [b)] $ y_0^{k,l}\in\Ri{k}(u)\colon M,y_0^{k,l}\models v_0^{k,l}$, 
        \item [c)] let $i\geq 0\colon y_{i+1}^{k,l}\in\Ri{k+1}^-(y_i^{k,l})\cap \Ri{k}(u)$; as seen just above $M,y_{i+1}^{k,l}\models v_{i+1}^{k,l}$ and, last, $M,y_{i+1}^{k,l}\models W_i^{k,l}$ (main induction hypothesis).  
        \end{itemize}
    \end{enumerate}
    All in all, $M,u\models u$ and  for all $\neg\Bi{k}\phi^k_l\in u$ $M,u\models W^{k,l}$
\end{proof}

We come now to complexity. Let $\spc$ be the function evaluating the amount of memory needed by a call to one of our functions. 
\begin{lemma}\label{polyspace}
$\sat(u)$ runs in polynomial space w.r.t.\ $\lgu$, i.e.\ $\spc(\sat(u))=R(\lgu)$ for some polynomial $R$ of constant degree.    
\end{lemma}
\begin{proof}
    First, we recall that functions \all\ and \algand\ are lazily evaluated. \\
    Obviously, \chooseCCS\ runs in polynomial space. 
    On another hand, as seen in lemma \ref{prolongation-to-infinite-window}, the size of each $(k,d(u),d)$-window for $u$ is bounded by $P(\lgu)$. Thus the functions \chooseW\ and \nextW\ run in polynomial space. It is also clear that functions $\sat$ and $\satW$ terminate since their recursion depth is bounded (respectively by $\lgu$ and by $\lgN$ the size of $N$) as well as their recursion width.
    Among all of these calls, let $\wpp$ be the argument with modal depth $<d(u)$ for which $\sat$ has the maximum cost in terms of space, i.e.\ such that $\spc(\sat(\wpp))$ is maximal. As well, among subwindows of $W$ (which are $(k,d(v),d)$-windows with $d(v)<d(u)$), let $W^+$ be the $(k,d(v^+),d)$-window, with $d(v^+)<d(u)$, for which the space used by $\satW(W^+,v^+,k,\uplim{v^+})$ is maximal. \\
    Let $W^0=\win{v}{W^0}{d(u)}$ be the $(k,d(u),d)$-window for $(u,v_0)$ chosen by $\chooseW$, and for $i\in\interf{1}{\uplim{u}}$, let $W^{j}$ be the $(k,d(u),d)$-windows for $(u,v_i)$ chosen by the successive recursive calls to $\nextW$. 
    Let us firstly evaluate the cost of $space(\satW(W^0,u,k,N))$. ). The function $\satW$  keeps its arguments in memory during the calls $\sat(s_0)$ and $\satW(W^0,u,k,N)$, then forget them and continue with $\satW(W^1,u,k,N-1)$. Let $\tau=\lgx{W^0}+\lgu+\lgx{k}+\lgx{\uplim{u}}$, we have $\tau\leq 4.\lgx{W^+}$. Now, let us proceed: \\
    
    $\spc(\satW(W^0,u,k,\uplim{u}))$
    \begin{longtable}{lll}
    $\leq$& $\max \{$&$\tau+\spc(\sat(v_0)),$\\
    &&$\tau+\spc(\satW(W_i^0,v_i^0,k+1,\uplim{v_i^0})),$\\
    &&$space(\satW(W^1,u,k,\uplim{u}-1))\}$\\
    $\leq$&$ \max \{$&$\tau+\spc(\sat(v^+)),$\\
    & & $\tau+\spc(\satW(W^+,v^+,k+1,\uplim{v^+})),$\\
    & & $space(\satW(W^1,u,k,\uplim{u}-1))\}$\\
    $\leq$ & $ \max \{$ & $\tau+\spc(\sat(v^+)),$\\
    & & $\tau+\spc(\satW(W^+,v^+,k+1,\uplim{v^+})),$\\
    & &$\max\{\tau+\spc(\sat(v^+)),$\\
    & & \hspace{0.8cm}$\tau+\spc(\satW(W^+,v^+,k+1,\uplim{v^+})),$\\ 
    & & $\hspace{0.8cm}\space(\satW(W^2,u,k,\uplim{u}-2))\}\}$\\
    $\leq$ & $\max \{ $ & $\tau+\spc(\sat(v^+)),$\\
    & & $\tau+\spc(\satW(W^+,v^+,k+1,\uplim{v^+})),$\\
    & & $\spc(\satW(W^2,u,k,\uplim{u}-2))\}$\\
    $\vdots$ & & \\
    $\leq $ &  $ \max \{ $ & $ \tau+\spc(\sat(v^+)) ,$\\
    & & $ \tau+\spc(\satW(W^+,v^+,k+1,\uplim{v^+})), $\\
    & & $ \spc(\satW(W^{\uplim{u}},u,k,0))\} $ \\
     $ \leq $ & & $ \tau+\max \{\spc(\sat(v^+)),$\\
    & & $ \hspace{1.5cm}\spc(\satW(W^+,v^+,k+1,\uplim{v^+}))\}$\\
      $ \leq $ & & $ \tau+\max \{\spc(\sat(v^+)),$\\
    & &  \hspace{1.5cm} $\tau+\max \{\spc(\sat(v^+)),$\\
    & & \hspace{3cm} $\spc(\satW(W^+,v^+,k+2,\uplim{v^+}))\}\}$\\
      $ \leq $ & & $ 2.\tau+\max \{\spc(\sat(v^+)),$\\
    & & \hspace{1.6cm} $\spc(\satW(W^+,v^+,k+2,\uplim{v^+}))\}$\\
    $\vdots$ &   & \\
      $ \leq $ & & $ \pi.\tau+\max \{\spc(\sat(v^+)),$\\
    & & \hspace{1.6cm} $\spc(\satW(W^+,v^+,\pi,\uplim{v^+}))\}$\\
       $ \leq $ & & $ \pi.\tau+\spc(\sat(v^+))$\\
    \end{longtable}
    
    \noindent Now, concerning the function $\sat$, it also keeps of its argument in memory during recursion in order to range over its $\Diamond$-formulas. obviously, in general, calls $\satW$ need more space than $\sat$ calls. For $\neg\Bi{k}\phi\in u$, let $W^{0,\neg\Bi{k}\phi}$ be the $(k,d(u),d)$-window chosen by $\chooseW(u,\chooseCCS(\{\neg\phi\}\cup\Bim{k}(u)),k), u,k,\uplim{u})$ (it exists, otherwise the algorithm stops). 
    Thus:\\
    $\spc(\sat(u))$\\
$\begin{array}{ll}
     &  \leq \lgu+\max \{\spc(\satW(W^{0,\neg\Bi{k}\phi},u,k,\uplim{u}))\colon\neg\Bi{k}\phi\in u\}\\
    & \leq \lgu+\pi.\tau+\spc(\sat(v^+))
    \end{array}$
    
With respect to the modal depth of the arguments (and since both $d(v^+)< d(u)$ and $\lgx{v^+}<c_{\mathrm{ccs}}.\lgu$) we are left with a recurrence equation of the form: 
    $\spc(d(u))\leq c_{\mathrm{ccs}}.\lgu+\pi.\tau+\spc(\sat(d(u)-1))$ with $\spc(0)=1$; if we set $c'=c_{\mathrm{ccs}}.\pi$ and since $\tau\geq \lgu$, this yields $\spc(\sat(d(u))\leq c'.\tau.(c'.\tau+c'.\tau)\leq 2.(c'.4.\lgx{W^+})^2$. From lemma \ref{prolongation-to-infinite-window}, we know that $\card({W^+})=Q(\lgu)$ (for a polynomial $Q$ of degree $\pi+1$), with elements of size bounded by $c_{\SF(u)}.\lgu$, hence $\lgx{W^+}=c_{\SF(u)}.\lgu.Q(\lgu)$ and $\spc(\sat(u))={\cal O}(R(u))$ for some polynomial $R$ of degree $2\pi+4$. 
\end{proof}
\begin{theorem}
    $DP_{\Kpi}$ is $\PSPACE$-complete. 
\end{theorem}
\begin{proof}
    On the one hand, $DP_{\Kpi}$ is $\PSPACE$-hard as recalled p.\ \pageref{pspace-hardness}, and on the other hand, function $\sat$ can decide non-deterministically and within polynomial space whether a set of formulas is satisfiable, hence satisfiability is in $\NPSPACE$, i.e.\ in $\PSPACE$ (by Savitch' theorem), and so $DP_{\Kpi}$ is in co-$\PSPACE$, i.e.\ in $\PSPACE$. 
\end{proof}

\section{Consequence on $K+\Diamond p\rightarrow\Diamond\Diamond p$}\label{KDe}
In \cite{BalGasq25}, the monomodal logic $\Kde=K+\Diamond p\rightarrow\Diamond\Diamond p$ has been proved to be both $\PSPACE$-hard and in $\EXPTIME$. Up to day, its exact satisfiability problem is still open. We will not close it, but will just add a little light on it, thought not very surprising. 

In the sequel, we suppose $\Pi=\nat$.\\ 

The algorithm may be easily adapted. But, before, let us go back to the proof of lemma \ref{prolongation-to-infinite-window} where we naturally assumed that $d(u)\geq \pi$. It is no longer the case and we must change the point of view for determining the cardinal $s(d(u))$ of a window for $(u,v_0)$ since only the decrease of the modal degree along subwindows will lead to an end of recursion. 

Let us have a look at the cardinal $s(d(u))$ of $(k,d(u),d)$-windows for $(u,v_0)$: 
$$\begin{array}{ll}
         s(d(u)) &  = (\mbox{if } n=0 \mbox{ then } 1 \mbox{ else }{d(u)}+\Sigma_{i=0}^{d(u)-1}s(d(u)-1)\\
         & \leq (\mbox{if } n=0 \mbox{ then } 1 \mbox{ else } d(u) + d(u).s(d(u)-1)\\
         &\leq d(u) + d(u)^2(u)+ d(u)^2.s(d(u)-1)\\
         &\cdots\\
         &\leq d(u) + d(u)^2(u)+ \cdot +d(u)^{d(u)}(u)+ d(u)^{d(u)}.s(0)\\
         &\leq f(d(u)) \hfill\mbox{\indent with $f(d(u))=2.d(u)^{d(u)+1}$}\\
    \end{array}$$
And the size $\lgx{W}$ of a $(k,d(u),d)$-window, is bounded by $c_{\SF}.\lgu.f(d(u))$. This changes the definition of $\uplim{u}$ but not the algorithm. 

When $d(u)=1$ and for each $\neg\Box\phi\in u$, we just need to consider if there exists $v\in\CCS(\Box^{\,\mhyphen}\cup\{\neg\phi\})$, then for the sake of completeness we would need to add $(v,v)$ to $R$ in order to guaranty that our model is dense (this is harmless for the truth lemma since $d(v)=0$). Then soundness and completeness proofs transfer straightforwardly, at the supplementary condition that in the code of $\satW$ we change ``${\tt or~k=\pi}$'' into ``${\tt or~d(u)=1}$''. Then functions $\sat$ and $\satW$ become (even if the parameter $k$ is now useless):
\setlength{\textfloatsep}{0pt}
\setlength{\floatsep}{0pt}
\begin{algorithm}[H]
 \floatname{algorithm}{Function}
\begin{algorithmic}
\Function{\sat}{$u$}:
\State {return}
\State {\hspace{0.3cm}$u\neq \{\bot\}$}
\State {\hspace{0.1cm}\algand\ }
\State {\hspace{0.2cm}$\all \{\satW(\chooseW(u,\chooseCCS(\{\neg\phi\}\cup\Box^{\,\mhyphen}(u)),k), u,k,\uplim{u})\colon\neg\Box\phi\in u\}$}
\EndFunction
\end{algorithmic}
\end{algorithm}

\begin{algorithm}[H]
\floatname{algorithm}{Function}
\begin{algorithmic}
\Function{\satW}{$W$,$u$,$k$,$N$}:\algorithmiccomment{$W$ is $\win{v}{W}{d(u)}$
}
\If {$N=0$ \textbf{or} $d(u)=1$} \algorithmiccomment{or $\langle (\{\bot\})_{i\in[0:d(u)]}, \emptyset\rangle$ if $\chooseW$}
\State {return \true}    \algorithmiccomment{or $\nextW$ has failed}
\Else 
\State {return}
\State {\hspace{0.83cm} \sat$(v_0)$}
\State {\hspace{0.2cm} $\algand\ \satW(W_0,v_{1},k+1,\uplim{v_1})$}\
\State {\hspace{0.2cm} \algand\ \satW(\nextW$(W,u,k),u,k,N-1))$}
\EndIf
\EndFunction
\end{algorithmic}
\end{algorithm}
Now, we must adapt the proof of lemma \ref{polyspace}. 

\begin{lemma}
    The space needed for the algorithm modified as said above for $\Pi=\nat$ is ${\cal O}(\lgu.Q(d(u)))$ where $Q$ is a polynomial of degree $d(u)+1$. 
\end{lemma}
\begin{proof}
    Among all of these calls, let $w^{d(u)}$ be the argument with modal depth $d(u)-1$ for which $\sat$ has the maximum cost in terms of space. As well, among subwindows of $W$, let $W^{d(u)}$ be the $(k,d(v^{d(u)}),d)$-window, with $d(v^{d(u)})=d(u)-1$, for which the space used by $\satW(W^{d(u)},v^{d(u)},k,\uplim{v^{d(u)}})$ is maximal. The inequalities are anyway similar and leads to:\\ 
    
        $\spc(\satW(W^0,u_{d(u)},k,\uplim{u}))$
    \begin{longtable}{lll}
    $\leq$& $\max \{$&$\tau+\spc(\sat(v^{d(u)-1})),$\\
    &&$\tau+\spc(\satW(W^{d(u)-1},v^{d(u)-1},k+1,\uplim{v^{d(u)-1}})),$\\
    &&$space(\satW(W^{d(u)-1},u,k,\uplim{u}-1))\}$\\
    $\vdots$ & & \\
    $\leq $ &  $ \max \{ $ & $ \tau+\spc(\sat(v^{d(u)-1})) ,$\\
    & & $ \tau+\spc(\satW(W^{d(u)-1},v^{d(u)-1},k+1,\uplim{v^{d(u)-1}})), $\\
    & & $ \spc(\satW(W^{d(u)-1},u,k,0))\} $ \\
     $ \leq $ & & $ \tau+\max \{\spc(\sat(v^{d(u)-1})),$\\
    & & $ \hspace{1.5cm}\spc(\satW(W^{d(u)-1},v^{d(u)-1},k+1,\uplim{v^{d(u)-1}}))\}$\\
      $ \leq $ & & $ \tau+\max \{\spc(\sat(v^{d(u)-1})),$\\
    & &  \hspace{1.5cm} $\tau+\max \{\spc(\sat(v^{d(u)-2})),$\\
    & & \hspace{3cm} $\spc(\satW(W^{d(u)-2},v^{d(u)-2},k+2,\uplim{v^{d(u)-2}}))\}\}$\\
      $ \leq $ & & $ 2.\tau+\max \{\spc(\sat(v^{d(u)-1})),$\\
    & & \hspace{1.6cm} $\spc(\satW(W^{d(u)-2},v^{d(u)-2},k+2,\uplim{v^{d(u)-2}}))\}$\\
    $\vdots$ &   & \\
      $ \leq $ & & $ d(u).\tau+\max \{\spc(\sat(v^{d(u)-1})),$\\
    & & \hspace{1.6cm} $\spc(\satW(W^0,v^0,k+d(u),\uplim{v^0}))\}$\\
       $ \leq $ & & $ d(u).\tau+\max \{\spc(\sat(v^{d(u)-1}))$\\
    \end{longtable}
    Now, w.r.t.\ the degree of $u$ we have: \\
    $\spc(\sat(u))$\\
$\begin{array}{ll}
     &  \leq \lgu+\max \{\spc(\satW(W^{0,\neg\Bi{k}\phi},u,k,\uplim{u}))\colon\neg\Bi{k}\phi\in u\}\\
    & \leq \lgu+d(u).\tau+\spc(\sat(v^{d(u)-1}))\\
    & \leq 2.\tau^2+\spc(\sat(v^{d(u)-1}))\\
    & \leq 2.\tau^2.d(u) \leq 2.\tau^3\leq 8.\lgx{W^{d(u)}}^3\\
    & \leq c''.\lgu^3.(d(u))^{3.d(u)+3} 
    \end{array}$
\end{proof}

Hence, the exponential need of memory of the algorithm is due to the modal depth of the formulas to be checked. From \cite{BalGasq25}, the satisfiability problem has been proved  to be in $\EXPTIME$ by means of selective filtration, and the above algorithm only provides an $\EXPSPACE$ bound. But in the classification of parameterized complexity (see \cite{FlumeGrohe06} for example), the above result allows us to state that the satisfiability problem for $\Kde$ is moreover in para-$\PSPACE$: if $d(u)$ is considered as a fixed parameter, then it is in $\PSPACE$, this cannot be seen in the selective filtration. This may provide a clue for determining its exact complexity. 

\section*{Conclusion}\label{conclusion}
With \emph{recursive windows}, we designed a tool for algorithms for bounded weakly-dense logics and the study of  their complexity. Hopefully they could as well prove useful beyond the strict scope of bounded weak-density for solving open problems: complexity of logics containing several more general expansion axioms of the form $\Diamond_{a} p\rightarrow \Diamond_{a_1}\ldots\Diamond_{a_m} p$ where $a=a_i$ for some $i$) or even when $a$ occurs several times in $\{a_1,\cdots a_m\}$. The complexity of these logics remain an open problem.

\end{document}